\documentclass{amsart}

\usepackage[latin1]{inputenc}
\usepackage{amsmath, amsthm, amssymb, booktabs, multirow, graphicx, booktabs}
\usepackage{dialogue}
\usepackage[stable]{footmisc}
\usepackage{tikz,nicefrac}
\usetikzlibrary{calc,matrix,arrows,decorations.markings}
\usepackage[colorlinks]{hyperref}
\usetikzlibrary{
  graphs,
  graphs.standard
}

\newtheorem{defin}{Definition}[section]

\newtheorem{theorem}[defin]{Theorem}
\newtheorem{proposition}[defin]{Proposition}
\newtheorem{lemma}[defin]{Lemma}

\newtheorem{example}[defin]{Example}

\newcommand{\R}{\mathbb{R}}

\newcommand{\Z}{\mathbb{Z}}

\DeclareMathOperator{\supp}{supp}
\DeclareMathOperator{\Vor}{Vor}

\newcommand{\optprob}[1]{{\arraycolsep=0pt%
  \begin{array}{r@{\ }l@{\quad}l}
    #1
  \end{array}}}

\begin{document}

\title{A polynomial time algorithm for solving the closest vector
  problem in zonotopal lattices}

\author{S. Thomas McCormick}
\address{S.T. McCormick, Sauder School
  of Business, University of British Columbia, Vancouver, Canada}
\email{tom.mccormick@sauder.ubc.ca}

\author{Britta Peis}
\address{B. Peis, School of Business and
  Economics, RWTH Aachen, 52072 Aachen, Germany}
\email{britta.peis@oms.rwth-aachen.de}

\author{Robert Scheidweiler}
\address{R. Scheidweiler, Paderborn, Germany}
\email{scheidweiler@math2.rwth-aachen.de}

\author{Frank Vallentin}
\address{F. Vallentin, Department Mathematik/Informatik, Universit\"at
  zu K\"oln, Weyertal 86--90, 50931 K\"oln, Germany}
\email{frank.vallentin@uni-koeln.de}

\date{June 18, 2021}

\subjclass{68Q25, 52C07}

\keywords{closest vector problem, lattice, Voronoi cell, zonotope,
  totally unimodular matrix, minimum mean cycle canceling}

\begin{abstract}
  In this note we give a polynomial time algorithm for solving the
  closest vector problem in the class of zonotopal lattices. The
  Voronoi cell of a zonotopal lattice is a zonotope, i.e.\ a projection of a
  regular cube. Examples of zonotopal lattices include lattices of
  Voronoi's first kind and tensor products of root lattices of type
  $\mathsf{A}$.  The combinatorial structure of zonotopal lattices can
  be described by regular matroids/totally unimodular matrices. We
  observe that a linear algebra version of the minimum mean cycle
  canceling method can be applied for efficiently solving the closest
  vector problem in a zonotopal lattice if the lattice is given as the
  integral kernel of a totally unimodular matrix.
\end{abstract}

\maketitle

\markboth{S.T.~McCormick, B.~Peis, R.~Scheidweiler, F.~Vallentin}{A
  polynomial time algorithm for solving CVP in zonotopal lattices}

\section{Introduction}

A \emph{lattice} $L$ of rank $r$ is a discrete subgroup of
$(\mathbb{R}^m, +)$ which spans a linear subspace of
dimension~$r$. One can specify a lattice by a \emph{lattice
  basis}; these are $r$ linearly independent vectors
$b_1, \ldots, b_r \in L$ so that $L$ is given by all their integral
linear combinations:
\[
L = \left\{ \sum_{i=1}^r \alpha_i b_i : \alpha_1, \ldots, \alpha_r \in
\mathbb{Z}\right\}.
\]

The central computational problems for lattices are the shortest
vector problem (SVP) and the closest vector problem (CVP). They have
many applications in mathematics, computer science, and engineering,
in particular in complexity theory, cryptography, information theory,
mathematical optimization, and the geometry of numbers; see for
instance \cite{Nguyen2009a}.

Solving the shortest vector problem amounts to finding a shortest
nonzero vector in a given lattice. In this paper we are concerned with
the \emph{closest vector problem} (CVP): Given a lattice
basis of $b_1, \ldots, b_r$ of $L$ and given a target vector
$t \in \mathbb{R}^m$ find a lattice vector $u \in L$ which is closest
to $t$, i.e.\
\[
\text{determine } u \in L \text{ with } \|u - t\| = \min_{v \in L} \|v - t\|,
\]
where $\|x\|$ denotes the standard Euclidean norm of a vector
$x \in \mathbb{R}^m$.  Without loss of generality, after performing an
orthogonal projection, we may assume that the target vector $t$ lies
in the span of $L$, which we denote by $\mathcal{L}$.

One can interpret CVP geometrically via the \emph{Voronoi cell} of the
lattice $L$ which is defined as
\[
\mathcal{V}(L) = \{x \in \mathcal{L} : \|x \| \leq \|v - x\| \text{ for
  all } v \in L\}.
\]
The Voronoi cell of $L$ is a polytope which tessellates the space
$\mathcal{L}$ by lattice translates $v + \mathcal{V}(L)$ for
$v \in L$. Now the CVP asks for a lattice vector $u$ so that the
target vector $t$ lies in $u + \mathcal{V}(L)$.

\medskip

In the past the closest vector problem has been studied
intensively. Here we only discuss results on algorithms and
complexity which are most relevant for us. We refer to
\cite{Aggarwal2015a} for an up-to-date discussion of the computational
complexity of CVP.

Van Emde Boas \cite{vanEmdeBoas1981a} established the
$\mathrm{NP}$-hardness of exactly solving CVP. Dinur, Kindler, Raz,
Safra \cite{Dinur2003a} showed that approximating CVP within a factor
of $r^{c/\log \log r}$, for some positive constant $c$, is
$\mathrm{NP}$-hard as well. Aharonov and Regev \cite{Aharonov2005a}
showed that approximating CVP within a factor of $\sqrt{r}$ lies in
$\mathrm{NP} \cap \mathrm{coNP}$.

On the algorithmic side, Micciancio and Voulgaris
\cite{Micciancio2013a} developed a deterministic algorithm for exactly
solving CVP which runs in $\tilde{O}(2^{2r})$ time and needs
$\tilde{O}(2^r)$ space. This was improved by Aggarwal, Dadush, and
Stephens-Davidowitz \cite{Aggarwal2015a} who achieved a
$2^{r+o(r)}$-time and space randomized algorithm. Hunkenschr\"oder,
Reuland, Schymura \cite{Hunkenschroeder2019a} considered the
possibility to improve the space complexity of the algorithm by
Micciancio and Voulgaris if one has a compact representation of the
lattice' Voronoi cell.

\medskip

In this note we are concerned with the polynomial time solvability of
the closest vector problem restricted to a special class of lattices.

\medskip

That CVP can be solved in polynomial time for special classes of
lattices has been proved in the case of lattices of Voronoi's first
kind by McKilliam, Grant, and Clarkson \cite{McKilliam2014a} and in
the case of tensor products $\mathsf{A}_m \otimes \mathsf{A}_n$ of
root lattices of type $\mathsf{A}$ by Ducas and van Woerden
\cite{Ducas2018a}.

The main result of this paper unifies and extends these two cases. For
this we consider lattices whose Voronoi cell is a zonotope. Zonotopes
are defined as projections of cubes; all of their faces (of any
dimension) are centrally symmetric. All lattices up to dimension three
have a zonotope as Voronoi cell, but starting from dimension four on,
there are lattices which do not have this property, for example the
root lattice $\mathsf{D}_4$ whose Voronoi cell is the
$24$-cell. Indeed, the three-dimensional facets of the $24$-cell are
regular octahedra and their two-dimensional faces are regular triangles
and thus are not centrally symmetric.

We show that one can exactly solve CVP for zonotopal lattices in
polynomial time using an algorithm of Karzanov and McCormick
\cite{Karzanov1997a}. Their algorithm can be seen as a linear algebra
version of the minimum mean cycle canceling algorithm of Goldberg and
Tarjan \cite{Goldberg1989a} for finding a minimum-cost circulation in
a network.

The set up is a follows: A totally unimodular matrix
$M \in \{-1,0,+1\}^{n \times m}$, i.e.\ every minor of $M$ is either
equal to $-1$, $0$, or $1$, is given. We consider the lattice $L$ of
all integer points lying in the kernel of $M$; the matrix $M$ will be
part of the algorithm's input. Furthermore, a \textit{separable}
convex objective function $w : \mathbb{R}^m \to \mathbb{R}$ is
given. Separability means that for every $i \in [m]$ we have a convex
function $w_i : \mathbb{R} \to \mathbb{R}$ so that
\[
  w(x) = \sum_{i=1}^m w_i(x_i) \quad \text{for} \quad (x_1, \ldots,
  x_m) \in \mathbb{R}^m.
\]
Then, under some technical conditions on the separable convex
objective function $w$, one can compute in
polynomial time a lattice vector $v \in L$ so that $w(v)$ is as small
as possible.

Since the work of Coxeter \cite{Coxeter1962a}, Shephard
\cite{Shephard1974a} and McMullen \cite{McMullen1975a} it is known
that the combinatorial structure of zonotopes which tile space by
translations is determined by a regular matroid and thus it is related to
totally unimodular matrices.

Zonotopal lattices are defined in Section~\ref{sec:zonotopal} and the
relation to regular matroids is reviewed. We develop the theory in
such a way that the separability of the objective function which
solves the CVP in this setting becomes apparent. We show that lattices
of Voronoi's first kind and that tensor product lattices
$\mathsf{A}_n \otimes \mathsf{A}_m$ are zonotopal lattices.

In Section~\ref{sec:algorithm} we discuss the algorithm of Karzanov
and McCormick. We cast the CVP for zonotopal lattices into a separable
convex optimization problem and verify that the technical conditions on the
separable convex objective function are fulfilled to ensure the
polynomial time solvability.

\section{Zonotopal lattices}
\label{sec:zonotopal}

In this section we collect basic definitions and facts about zonotopal
lattices. Zonotopal lattices were first defined by
Gerritzen~\cite{Gerritzen1982a} when he gave a metric to Tutte's
regular chain groups (see for example Tutte~\cite{Tutte1971a}). The
theory of zonotopal lattices was further developed by
Loesch~\cite{Loesch1990a} and Vallentin~\cite{Vallentin2000a},
\cite{Vallentin2003a}, \cite{Vallentin2004a}.

Space tiling zonotopes have been thoroughly investigated in the
literature: Main examples of zonotopal lattices are the lattice of
integral flows and the lattice of integral cuts on a finite graph
which were considered by Bacher, de la Harpe, Nagnibeda
\cite{Bacher1997a}. Lattices whose Voronoi cell are zonotopes can be
dually interpreted by Delone subdivisions and hyperplane arrangements;
this has been done by Erdahl and Ryshkov \cite{Erdahl1994a} who
developed the theory of lattice dicings for this. Zonotopes which tile
space by translations were studied by Coxeter \cite{Coxeter1962a},
Jaeger \cite{Jaeger1983a}, Shephard \cite{Shephard1974a}, and McMullen
\cite{McMullen1975a}, see also \cite{Bjoerner1993a}.

\subsection{Combinatorics: Regular chain groups, regular matroids,
  totally unimodular matrices}
\label{ssec:combinatorics}

We start by briefly recalling fundamental definitions and results of
Tutte's theory of regular chain groups. Chain groups are defined over
general integral domains $R$ (commutative rings with a unit
element and no divisors of zero). In this paper we only need
$R = \mathbb{R}$ or $R = \mathbb{Z}$. So we sometimes simplify Tutte's
original notation. Regular chain groups are closely related to regular
matroids and totally unimodular matrices. We refer, for example, to
Camion \cite{Camion2006a}, Oxley \cite{Oxley2011a}, Schrijver
\cite{Schrijver1986a}, Tutte \cite{Tutte1965a}, \cite{Tutte1971a},
Welsh \cite{Welsh1976a} for proofs and more details.

\smallskip

Let $\mathcal{L}$ be a subspace of $\mathbb{R}^m$. The \emph{support}
of a vector $x = (x_1, \ldots, x_m) \in \mathcal{L}$ is given by
\[
\supp x = \{i  \in [m]: x_i \neq 0\} \quad \text{with} \quad [m] =
\{1, \ldots, m\}.
\]
A non-zero vector $x \in \mathcal{L}$ is called an \emph{elementary chain} if
it has minimal (inclusion-wise) support among all non-zero vectors in
$\mathcal{L}$.  An elementary chain $x$ is called a \emph{primitive chain} if
$x_i \in \{-1,0,+1\}$ for all $i \in [m]$. A subspace $\mathcal{L}$ is
called \emph{regular} if every elementary chain is a multiple of a
primitive chain.

The set of supports of elementary chains in a regular subspace forms
the circuits of a regular matroid, a matroid which is representable
over every field.  If a matrix $M$ is totally unimodular, then the kernel of
$M$ is a regular subspace. Conversely, every regular subspace can be
represented as kernel of a totally unimodular matrix.

\smallskip

The orthogonal complement of a regular subspace $\mathcal{L}$ which is
defined by
\[
\mathcal{L}^\perp = \left\{y \in \mathbb{R}^m : \sum_{i = 1}^m x_i y_i = 0
\text{ for all } x \in \mathcal{L}\right\}
\]
is again regular.

Let $S$ be a subset of $[m]$. We define the \emph{deletion}
$\mathcal{L} \setminus S$ by
\[
  \mathcal{L} \setminus S = \left\{(x_i)_{i \not\in S} : x = (x_1,
    \ldots, x_m) \in \mathcal{L}, \; \supp x \cap S = \emptyset\right\},
\]
and the \emph{contraction} $\mathcal{L} / S$ by
\[
  \mathcal{L} / S = \left\{(x_i)_{i \not\in S} : x = (x_1, \ldots,
    x_m) \in \mathcal{L}\right\}.
\]
Both operations preserve regularity. We say that a subspace is a
\emph{minor} of $\mathcal{L}$ if it is obtained from $\mathcal{L}$ by
a sequence of deletions and contractions.

\smallskip

Two main examples of regular subspaces come from directed graphs. Let
$D = (V, A)$ be an acyclic, directed graph with vertex set $V$ and arc
set $A$. By $M(D) \in \{-1,0,+1\}^{V \times A}$ we denote the
vertex-arc incidence matrix of $D$ which is a totally unimodular
matrix. Define the regular subspace $\mathcal{L}(D)$ as the kernel of
$M$:
\[
\mathcal{L}(D) = \left\{x \in \mathbb{R}^A : M(D)x = 0\right\}.
\]
The primitive chains of $\mathcal{L}(D)$ correspond to the simple
circuits/cycles (forward and backward arcs are allowed) of the directed graph
$D$. Regular subspaces which can be realized by this construction are
called \emph{graphic}.  The primitive chains of the orthogonal complement
$\mathcal{L}(D)^\perp$ correspond to the simple cuts/bonds (forward and
backward arcs are allowed) of $D$. Such a regular subspace is called
\emph{cographic}. Minors of graphic (resp.\ cographic) subspaces are
graphic (resp.\ cographic). The dimension of $\mathcal{L}(D)$ equals
$|A| - |V| + k$ where $k$ is the number of connected components of the
underlying undirected graph and the dimension of
$\mathcal{L}(D)^\perp$ is $|V|-k$.

\smallskip

Tutte \cite{Tutte1958ab} gave a characterization of graphic and
cographic subspaces in terms of forbidden minors. For this let
$\mathsf{K}_m$ be the complete graph on $m$ vertices and let
$\mathsf{K}_{m,n}$ the complete bipartite graph where one partition
has $m$ vertices and the other one has $n$ vertices. Tutte showed that
a regular subspace is graphic if and only if it contains neither
$\mathcal{L}(\mathsf{K}_5)^\perp$ nor
$\mathcal{L}(\mathsf{K}_{3,3})^\perp$ as minors. Dually, a regular
subspace is cographic if and only if it contains neither
$\mathcal{L}(\mathsf{K}_5)$ nor $\mathcal{L}(\mathsf{K}_{3,3})$ as
minors. The central structure theorem about regular subspaces is
Seymour's decomposition theorem \cite{Seymour1980a}: One may construct
every regular subspace as $1$-, $2$-, and $3$-sums of regular
subspaces starting from graphic, or cographic subspaces, 
or the special regular subspace called
$\mathsf{R}_{10} \subseteq \mathbb{R}^{10}$; see also Truemper
\cite{Truemper1992a}.

\subsection{Geometry: Strict Voronoi vectors, Voronoi cells}

A regular subspace $\mathcal{L}$ comes together with a \emph{regular
  lattice} $L = \mathcal{L} \cap \mathbb{Z}^m$. One can show, see
\cite[Chapter 1.2]{Tutte1971a}, that in a regular lattice every vector
$v \in L$ is a \emph{conformal sum} of primitive chains
$w_1, \ldots, w_s \in L$:
\begin{equation}
\label{eq:conformal}
v = w_1 + \cdots + w_s \; \text{ with } \; (w_j)_i (w_k)_i \geq 0 \; \text{ for all } \; i \in [m] \; \text{ and } \; j, k \in [s].
\end{equation}
When $\mathcal{L} \subseteq \mathbb{R}^m$ is a graphic (cographic)
subspace we call the associated lattice
$L = \mathcal{L} \cap \mathbb{Z}^m$ \emph{graphic} (\emph{cographic})
as well. The graphic lattices are the lattices of integral flows and
the cographic lattices are the lattices of integral cuts in the
framework of~\cite{Bacher1997a}.

We equip the space $\mathbb{R}^m$ with an inner product which is
defined by giving positive weights on the set $[m]$: For a positive
vector $g \in \mathbb{R}^m_{>0}$ define the inner product
\[
(x,y)_g = \sum_{i = 1}^m g_i x_i y_i.
\]
The standard basis vectors $e_1, \ldots, e_m$ form in this way an
orthogonal basis which does not need to be orthonormal.  A regular lattice $L$
with inner product $(\cdot, \cdot)_g$ is called a \emph{zonotopal
  lattice}. As we explain below, this terminology refers to the fact
that the Voronoi cell of a zonotopal lattice is a zonotope. The
\emph{Voronoi cell} of $L$ is
\[
\mathcal{V}(L) = \{x \in \mathcal{L} : (x,x)_g \leq (v-x, v-x)_g \text{ for all }
v \in L\},
\]
which is a centrally symmetric polytope. Lattice vectors $v \in L$
which determine a facet defining hyperplane
\[
H_v = \{x \in \mathcal{L} : (x,x)_g = (v-x,v-x)_g\} = \left\{x \in
  \mathcal{L} : (x,v)_g = \frac{1}{2} (v,v)_g\right\}
\]
of $\mathcal{V}(L)$ are called \emph{strict Voronoi vectors}
(sometimes also called \emph{relevant vectors}). We denote the set of all strict
Voronoi vectors by $\Vor(L)$.

Voronoi showed (see for example \cite[Chapter~21,
Theorem~10]{Conway1988a} or \cite{Conway1997a}), for arbitrary
lattices $L$, that a nonzero vector $v \in L$ is a strict Voronoi
vector if and only if $\pm v$ are the only shortest vectors in
$v + 2L$.

\medskip

In the following let $\mathcal{L} \subseteq \mathbb{R}^m$ be a regular
subspace and let $L = \mathcal{L} \cap \mathbb{Z}^m$ be the
corresponding regular lattice with positive vector
$g \in \mathbb{R}^m_{>0}$. Essentially, the arguments given below can
also be found in \cite{Dutour2009a} in the special case of cographic
lattices with constant $g$.

\smallskip

Applying Voronoi's characterization to $L$ yields:

\begin{proposition}
\label{strict Voronoi}
A lattice vector of $L$ is a primitive chain if and only if it is a strict Voronoi vector of $L$.
\end{proposition}

\begin{proof}
  Let $v \in L$ be a primitive chain and let $u \in v + 2L$ be a
  lattice vector with $u \neq \pm v$.  We have
  $v - u \in 2L \subseteq 2\Z^n$ and $v_i \in \{-1,0,+1\}$, for all
  $i \in [m]$, which shows $\supp v \subseteq \supp u$.  If
  $\supp v \neq \supp u$, then $(v,v)_g < (u,u)_g$.  If
  $\supp v = \supp u$, then there exists a factor
  $\alpha \in \Z \setminus \{-1,+1\}$ so that $u = \alpha v$, hence
  $(v,v)_g < (u,u)_g$.  In both cases $\pm v$ are the only shortest
  vectors in $v + 2L$.  Hence, $v$ is a strict Voronoi vector.

  Conversely, let $v \in L$ be a strict Voronoi vector. Write $v$ as a
  conformal sum of primitive chains as in~\eqref{eq:conformal}. Set
  $u = v - 2w_1$. Then
\[
(u,u)_g = (v,v)_g - 4 (v - w_1, w_1)_g \leq (v,v)_g,
\]
since $(v - w_1, w_1)_g \geq 0$ by~\eqref{eq:conformal}. Hence,
$\pm v$ is the unique shortest vector in the coset $v + 2L$ if and
only if $s = 1$.
\end{proof}

The following special case of Farkas lemma is proved e.g.\ in
\cite[Theorem 22.6]{Rockafellar1970a}.

\begin{lemma}
\label{farkas}
Let $x \in \R^m$ be a vector, and let
$\alpha_1, \ldots, \alpha_m \in \R \cup \{\pm\infty\}$.  Exactly one
of the following two alternatives holds:
\begin{enumerate}
\item There exists a vector $y'$ with $(y',z)_g = 0$ for all
  $z \in \mathcal{L}$ so that
\[
y' \in x + \prod_{i=1}^m [-\alpha_i, \alpha_i].
\]
\item There exists a vector $y \in \mathcal{L}$ such that
\[
\text{for all } \; z \in x + \prod_{i=1}^m [-\alpha_i, \alpha_i] \; \text{ we have } \; (y,z)_g > 0.
\]
\end{enumerate}
If the second condition holds, then one can choose $y$ to be
a primitive chain of $L$.
\end{lemma}

\begin{theorem}
\label{projection}
Let $\pi_g : \mathbb{R}^m \to \mathcal{L}$ be the orthogonal
projection of $\R^m$ onto $\mathcal{L}$. Then,
$\mathcal{V}(L) = \pi_g([-1/2,1/2]^m)$.
\end{theorem}

\begin{proof}
  For a vector $x \in [-1/2,1/2]^m$ inequality
  $(x,v)_g \leq \frac{1}{2} (v,v)_g$ holds for all
  $v \in \Z^m \setminus \{0\}$. Decompose $x$ orthogonally
  $x = y + y'$ with $y = \pi_g(x) \in \mathcal{L}$. For all
  $v \in L \setminus \{0\}$ we have
\[
(y,v)_g= (x,v)_g - (y',v)_g = (x,v)_g \leq
\frac{1}{2} (v,v)_g.
\]
Thus, $\pi_g(x) \in \mathcal{V}(L)$.

Let $x \in \mathcal{V}(L)$ be a vector of the Voronoi cell. If there
exists $y'$ with $(y',z)_g = 0$ for all $z \in \mathcal{L}$ so that
$y' \in -x + [-1/2, 1/2]^m$, then $x + y' \in [-1/2,1/2]^m$ and
$\pi_g(x + y') = \pi_g(x) = x$. Suppose that such a vector $y'$ does
not exist. Then by Lemma~\ref{farkas} there is a primitive chain $v \in L$ so that
\[
(v, -x + [-1/2,1/2]^m)_g > 0.
\]
This implies $(v,-x - \frac{1}{2} v)_g > 0$. Hence,
$-x \not\in \mathcal{V}(L)$; a contradiction because $\mathcal{V}(L)$
is centrally symmetric.
\end{proof}

This theorem proves that the Voronoi cell of a zonotopal lattice is
indeed a zonotope. The operations deleting or contracting correspond
to contracting the corresponding zones or projecting along the
corresponding zones of $\mathcal{V}(L)$, as mentioned in
\cite[Proposition 2.2.6]{Bjoerner1993a}. Also the combinatorial
structure of $\mathcal{V}(L)$, which is independent of $g$, is
completely encoded in the covectors of the oriented matroid defined by
$\mathcal{L}$, see \cite[Proposition 2.2.2]{Bjoerner1993a}.

Conversely, Erdahl \cite{Erdahl1999a}, see also \cite{Vallentin2000a},
\cite{Vallentin2004a} for an alternative proof, showed that every
zonotope which tiles spaces by translates is the affine linear image
of the Voronoi cell of a zonotopal lattice.

\subsection{Example: Lattices of Voronoi's first kind}

McKilliam, Grant, and Clarkson \cite{McKilliam2014a} gave a polynomial
time algorithm for solving the closest vector problem for lattices of
Voronoi's first kind (with known obtuse superbasis, see below). Now we
show that these lattices correspond to cographic lattices.

Following Conway and Sloane \cite{Conway1992a} we say that a lattice
$L$ is of \textit{Voronoi's first kind} if $L$ has an \textit{obtuse
  superbasis}: These are $n+1$ vectors $b_0, b_1, \ldots, b_n$ so that
the following three conditions hold:
\begin{enumerate}
\item[(i)] $b_1, \ldots, b_n$ is a basis of $L$,
\item[(ii)] $b_0 + b_1 + \cdots + b_n = 0$,
\item[(iii)] $b_i^{\sf T} b_j \leq 0$ for $i,j = 0, \ldots, n$ and $i \neq j$.
\end{enumerate}
A classical theorem of Voronoi states that every lattice in dimensions
$2$ and $3$ has an obtuse superbasis, see Conway and Sloane
\cite[Section 7]{Conway1992a}. However, starting from dimension
$n = 4$ on, not every lattice is of Voronoi's first kind.

In the setting of zonotopal lattices, lattices of Voronoi's first kind
appear as cographic lattices: Let $L \subseteq \mathbb{R}^n$ be a
lattice of Voronoi's first kind having an obtuse superbasis
$b_0, \ldots, b_n$. Define the directed graph $D = (V, A)$ with vertex
set $V = \{b_0, \ldots, b_n\}$ where we draw an arc $a_{ij}$ between
vertices $b_i$ and $b_j$ whenever $b_i^{\sf T} b_j < 0$ and $i <
j$. We assign to the arc $a_{ij}$ the (positive) weight
$g_{ij} = -b_i^{\sf T} b_j$.

The undirected graph which underlies $D$ is called \textit{Delone
  graph} of $L$, see \cite{Conway1992a}. In fact, the choice of the
directions of the arcs is arbitrary, as long as the graph does not
contain a directed cycle.

\begin{proposition}
The cographic lattices are exactly the lattices of Voronoi's first kind.
\end{proposition}

\begin{proof}
  The graph $D$ is weakly connected (i.e.\ the underlying undirected graph
  is connected) since $L$ has rank $n$: For suppose not. Then one can
  partition the vertex set $V = V_1 \cup V_2$ so that there is no arc
  between $V_1$ and $V_2$. Consider
  the spaces $U_i$ spanned by the vectors in $V_i$, with $i = 1, 2$. These spaces are
orthogonal and we have $\dim U_i = |U_i|$ because of (i)
  and (ii). Hence, $\dim(U_1 + U_2) = n+1$,  contradicting that the
  rank of $L$ is $n$.

  Consider the vertex-arc incidence matrix
  $M(D) \in \{-1,0,+1\}^{V \times A}$ of $D$ and let
  $v_0, \ldots, v_n \in \mathbb{R}^A$ be the row vectors of $M$. Their
  integral span coincides with the cographic lattice
  $L' = \mathcal{L}(D)^\perp \cap \mathbb{Z}^A$. Furthermore, the
  vectors $v_0, \ldots, v_n$ form an obtuse superbasis of $L'$ and
  $(v_i,v_j)_g = -g_{ij} = b_i^{\sf T} b_j$ holds when $v_i$ and $v_j$
  are adjacent in $D$. Hence, the lattice $L$ which is of Voronoi's
  first kind is isometric to the cographic lattice $L'$.

  Clearly, this construction can be reversed. Starting from a
  vertex-arc incidence matrix of a weakly connected acyclic directed
  graph defining a cographic lattice one can get an obtuse superbasis
  of this lattice. If the graph defining the cographic lattice is not
  weakly connected, then one can make it weakly connected by
  identifying vertices of distinct connected components without
  changing the cographic lattice.
\end{proof}

For instance, the root lattice
\[
\mathsf{A}_n = \left\{x \in \mathbb{Z}^{n+1} : \sum_{i=1}^{n+1} x_i = 0\right\}
\]
is a lattice of Voronoi's first kind. Its Delone graph is the cycle
graph $C_{n+1}$ of length $n+1$. The dual lattice $ \mathsf{A}_n^*$ is
again a lattice of Voronoi's first kind. Its Delone graph is the
complete graph $\mathsf{K}_{n+1}$ on $n + 1$ vertices. The Voronoi
cell of $\mathsf{A}_n^*$ is the $n$-dimensional permutahedron.

\subsection{Example: Tensor product of root lattices of type $\mathsf{A}$}

Ducas and van Woerden \cite{Ducas2018a} gave a polynomial time
algorithm for solving the closest vector problem for tensor products
of the form $\mathsf{A}_m \otimes \mathsf{A}_n$.  Now we show that
these lattices correspond to the graphic lattices for the complete
bipartite graph $K_{m+1,n+1}$.

Let $L_1 \subseteq \mathbb{R}^{n_1}$ be a lattice of rank $r_1$ with
basis $a_1, \ldots, a_{r_1}$ and let $L_2 \subseteq \mathbb{R}^{n_2}$
be a lattice of rank $r_2$ with basis $b_1, \ldots, b_{r_2}$. Then
their \emph{tensor product} is the lattice
$L_1 \otimes L_2 \subseteq \mathbb{R}^{n_1 n_2}$ having basis
$a_i \otimes b_j$ with $i = 1, \ldots, r_1$ and $j = 1, \ldots, r_2$.

\begin{proposition}
  The tensor product lattice $\mathsf{A}_m \otimes \mathsf{A}_n$
  coincides with the graphic lattice of the complete bipartite graph
  $\mathsf{K}_{m+1,n+1}$.
\end{proposition}

\begin{proof}
  Recall that the Delone graph of $\mathsf{A}_m$ is the cycle graph
  $C_{m+1}$. A basis of $\mathsf{A}_m$ is
  \[
  b_1 = e_1 - e_2,\;  b_2 = e_2 - e_3, \; \ldots, \; b_m = e_m - e_{m+1},
\]
where $e_1, \ldots, e_{m+1}$ are the standard basis vectors of
$\mathbb{R}^{m+1}$. A basis of $\mathsf{A}_n$ is
$c_j = f_j - f_{j+1}$, with $j \in [n]$ where $f_1, \ldots, f_{n+1}$
are the standard basis vectors of $\mathbb{R}^{n+1}$. This defines the following
basis of $\mathsf{A}_m \otimes \mathsf{A}_n$
\begin{equation}
    \label{eq:basis}
  b_i \otimes c_j = e_i \otimes f_j - e_{i+1}
  \otimes f_j + e_{i+1} \otimes f_{j+1} - e_i \otimes f_{j+1} \quad \text{for } i \in [m],
  \; j \in [n].
  \end{equation}
  If one orients all the arcs $A$ of $\mathsf{K}_{m+1,n+1}$ from the
  left $m+1$ vertices to the right $n+1$ vertices then the
  basis~\eqref{eq:basis} lies in the graphic lattice
  $\mathcal{L}(\mathsf{K}_{m+1,n+1}) \cap \mathbb{Z}^A$ as it
  corresponds to the cycle with edges
  $(e_i,f_j), (f_j,e_{i+1}), (e_{i+1},f_{j+1}), (f_{j+1},e_i)$.  Since
  the dimension of the graphic space of the complete bipartite graph
  is
\[
\dim \mathcal{L}(\mathsf{K}_{m+1,n+1}) =  (m+1)(n+1) - (m+1+n+1) + 1 =
mn,
\]
we see that the basis~\eqref{eq:basis} also forms a basis of the graphic lattice.
\end{proof}

\begin{example}
  We give a basis of the graphic lattice
  $\mathsf{A}_2 \otimes \mathsf{A}_1$ corresponding to the complete
  bipartite graph $\mathsf{K}_{3,2}$.

\begin{minipage}{8cm}
\begin{eqnarray*}
 b_1 \otimes c_1 & = & e_1 \otimes f_1 -  e_2
                       \otimes f_1 + e_2 \otimes f_2  - e_1 \otimes f_2 \\
 b_2 \otimes c_1 & = & e_2 \otimes f_1 - e_3
  \otimes f_1 + e_3 \otimes f_2 - e_2 \otimes f_2 \\
\end{eqnarray*}
\end{minipage}
\hspace*{1cm}
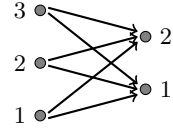
\begin{minipage}{2.5cm}
\begin{tikzpicture}[scale=0.7]
\small
\draw [fill=gray] (0,0) circle (0.10);
\draw [fill=gray] (0,1) circle (0.10);
\draw [fill=gray] (0,2) circle (0.10);
\draw [fill=gray] (2,0.5) circle (0.10);
\draw [fill=gray] (2,1.5) circle (0.10);

\node[left] at (-0.1,0) {$1$};
\node[left] at (-0.1,1) {$2$};
\node[left] at (-0.1,2) {$3$};
\node[right] at (2.1,0.5) {$1$};
\node[right] at (2.1,1.5) {$2$};

\draw[thick] [->] (0.15,-0.05) -- (1.85,0.4);
\draw[thick][->] (0.15,0.05) -- (1.85,1.4);

\draw[thick][->] (0.15,0.95) -- (1.85,0.5);
\draw[thick][->] (0.15,1.05) -- (1.85,1.5);

\draw[thick][->] (0.15,1.95) -- (1.85,0.6);
\draw[thick][->] (0.15,2.05) -- (1.85,1.6);
\end{tikzpicture}
\end{minipage}
\end{example}

\section{Minimum mean cycle canceling algorithm for CVP}
\label{sec:algorithm}

In this section we show how to derive the following theorem from the
results of Karzanov and McCormick~\cite{Karzanov1997a}; see also
\cite{Karzanov1995a} for the conference version.

\begin{theorem}
  Let $M \in \{-1,0,+1\}^{n \times m}$ be a given totally unimodular
  matrix and let $g \in \mathbb{Q}^m_{>0}$ be a given positive,
  rational vector. By $\mathcal{L}$ we denote the kernel of $M$ which
  is a regular subspace. This defines the zonotopal lattice
  $L = \mathcal{L} \cap \mathbb{Z}^m$ with inner product
  $(\cdot, \cdot)_g$. Let $t \in \mathcal{L} \cap \mathbb{Q}^m$ be a
  given rational (target) vector. Then one can compute a lattice vector
  $u \in L$ with $\|u-t\| = \min_{v \in L} \|v-t\|$ in polynomial
  time.
\end{theorem}

For the proof define the separable convex function
$w : \mathbb{R}^m \to \mathbb{R}$ by
\[
w_i(v_i) = g_i (v_i - t_i)^2 \quad \text{so that} \quad w(v) =
\sum_{i=1}^m g_i(v_i - t_i)^2 = (v - t, v - t)_g.
\]
Then solving the closest vector problem for $L$ given the target
vector $t$ amounts to finding a minimizer for $w(v)$ among all lattice
vectors $v \in L$. So we can apply the results of Karzanov and McCormick
to solve the closest vector problem for zonotopal lattices.

The minimum mean cycle canceling method gives a polynomial time
algorithm for solving the closest vector problem here. To see this we
have to verify some technical conditions for $w$ which we will do now.

We describe how the minimum mean cycle canceling method works in our
setting and discuss which arguments of the paper of Karzanov and
McCormick have to be applied to prove that the algorithm runs in
polynomial time.

We start by setting up notation. The \textit{(discrete) right
  derivative} of $w_i$ is
\[
c_i^+(v_i) = w_i(v_i + 1) - w_i(v_i) = g_i(v_i + 1 - t_i)^2 - g_i(v_i -
t_i)^2 = g_i( 2(v_i - t_i) + 1).
\]
Similarly the \textit{(discrete) left derivative} of $w_i$ is
\[
  c_i^-(v_i) = w_i(v_i) - w_i(v_i - 1) = g_i( 2(v_i - t_i) - 1).
\]
Technically we replace the quadratic objective function $w$ by its
piecewise linear approximation at lattice points.

The \textit{cost} of the strict Voronoi vector $u \in L$ at a lattice vector
$v$ is
\[
c(v,u) = \sum_{i \in u^+} c_i^+(v_i) - \sum_{i \in u^-} c_i^-(v_i),
\]
where
\[
  u^+ = \{i \in \supp u : u_i = +1\} \quad \text{and} \quad
  u^- = \{i \in \supp u : u_i = -1\}.
\]
If the cost $c(v,u)$ is negative, then $v+u$ is closer to $t$
than $v$ because we have 
 \[
 (v+u-t, v+u-t)_g = (v-t, v-t)_g + c(v,u),
\]
which is easily verified.

The \textit{mean cost} of $u$ at $v$ is
 \[
   \overline{c}(v,u) = \frac{c(v,u)}{|\supp u|}.
\]
A strict Voronoi vector $u$ is called a \textit{minimum mean strict
  Voronoi vector} for $v$ if its mean cost $\overline{c}(v,u)$ is a
small as possible. The following quantity is used to measure the
progress of the algorithm:
\[
\lambda(v) = \max\left\{0, -\min_{u \in \Vor(L)} \overline{c}(v,u)\right\},
\]
where $\Vor(L)$ denotes the set of strict Voronoi vectors where we
used Proposition~\ref{strict Voronoi}. Now \cite[Lemma
3.1]{Karzanov1997a} says that $v$ is a solution of the closest vector
problem if and only if $\lambda(v) = 0$. \cite[Proof of Lemma
3.2]{Karzanov1997a} shows that the following linear program computes
$-\lambda(v)$:
\[
  \optprob{\text{min} & \sum\limits_{i=1}^m (c_i^+(v_i) x^+_i - c_i^-(v_i) x^-_i) \\[0.5ex]
    & M(x^+ - x^-) = 0\\[0.5ex]
    & e^{\sf T} (x^+ + x^-) = 1\\[0.5ex]
    & x^+, x^- \in \mathbb{R}^m_{\geq 0}
  },
\]
where $e = (1, \ldots, 1) \in \mathbb{R}^m$ is the all-ones vector.
One can furthermore find a minimum mean strict Voronoi vector at $v$
by first determining $\lambda(v)$ and then finding a vector
$u = x^+ - x^-$ with minimal support with
$\lambda(v) = -\overline{c}(v,u)$. This can be done by solving at most
$m$ auxiliary linear programs where one greedily probes to set
coordinates to $0$, which is possible because every lattice vector is
a conformal sum of primitive chains~\eqref{eq:conformal}.

Now the minimum mean cycle canceling algorithm works as follows: We
start at the origin $v = 0$. As long as $\lambda(v)$ is positive, we
improve $v$, moving it closer to the target vector $t$ by finding a
minimum mean strict Voronoi vector $u$ at $v$ and updating $v$ to
$v + \varepsilon u$. The step size $\varepsilon$ is determined by
\cite[(16)]{Karzanov1997a} which is the minimum integer $\Delta$ so
that
\begin{equation}
\label{eq:delta}
\Delta \in \left[\frac{\lambda(v)}{g_i}, \frac{\lambda(v)}{g_i}+1 \right]
\quad \text{for } i \in [m].
\end{equation}
Indeed, for instance if $u_i = +1$, then the bounds for $\Delta$ in
\cite[(16)]{Karzanov1997a} are
\[
  c_i^-(v_i + \Delta) \leq c_i^+(v_i) + \lambda(v) \leq c_i^+(v_i + \Delta),
\]
yielding in our setting the interval in~\eqref{eq:delta} which
clearly always contains an integer.

Choosing the step size like this makes sure that
$\lambda(v+\varepsilon u) \leq \lambda(v)$, see \cite[Lemma
3.3]{Karzanov1997a}. By \cite[Lemma 3.4]{Karzanov1997a} we see that
after $m-n$ iterations the value of $\lambda$ decreases by a factor of
at most $(1-1/(2m))$, so that we have a geometric decrease.

We start with $v = 0$ and we assume that $\lambda(0) > 0$. Let $u$ be
a minimum mean strict Voronoi vector at $0$. If $g_i$ and $t_i$ are
rational, it is immediate to see that the binary encoding length of
$\lambda(0)$ is polynomial in the input size. If $g_i$ and
$t_i$ are rational, then we can also derive a stopping criterion for
the algorithm. Because of rationality, there exists an integer $K$ so
that $K c(v,u)$ is an integer for all $v \in L$ and all strict Voronoi
vectors $u$. The binary encoding length of $K$ is polynomial in the
input size. If $\lambda(v) < \delta$ for $\delta < \frac{1}{Km}$ then
$v$ is a closest vector to $t$ because from \cite[Proof of Lemma
6.1]{Karzanov1997a} it follows that
\[
  c(v,u) > -m \delta \geq -\frac{1}{K}
 \]
 and so $c(v, u)$ is nonnegative. The bound on $\lambda(0)$, the
 stopping criterion together with the geometric decrease of $\lambda$
 show that only a polynomial number of iterations are needed to find a
 closest vector.

\section{Further remarks and questions}

After the first version of this paper was submitted to the arXiv
repository, Rico Zenklusen pointed us to the paper~\cite{Hochbaum1990a} by
Hochbaum and Shanthikumar. With their method it is possible, together
with results of Artmann, Weismantel, and Zenklusen \cite{Artmann2017a}, to
efficiently solve the closest vector problem for a larger class of
lattices, for example for lattices of the form
\[
 L = \{x \in \mathbb{Z}^n : Mx = 0\}
\]
where the matrix $M$ is totally bimodular, which means that every
subdeterminant is bounded by $2$ in absolute value. The combinatorics
and the geometry of the Voronoi cells of these kind of lattices is
currently unexplored.

Another very interesting open question is if detecting that a given
lattice is isometric to a zonotopal lattice is computationally
easy. To the best knowledge of the authors even for the class of
cographical lattices this is not known.

\section*{Acknowledgements}

We thank Rico Zenklusen for his valuable remarks. We like to thank the
reviewers for their constructive suggestions which helped to improve the
presentation of the paper.
 
This project has received funding from the European Union's Horizon
2020 research and innovation programme under the Marie
Sk\l{}odowska-Curie agreement No 764759. The fourth named author is partially
supported by the SFB/TRR 191 ``Symplectic Structures in Geometry,
Algebra and Dynamics'' and by the project ``Spectral bounds in
extremal discrete geometry'' (project number 414898050), both funded
by the DFG.


\begin{thebibliography}{[99]}

\bibitem{Aggarwal2015a}
D.~Aggarwal, D.~Dadush, N. Stephens-Davidowitz,
\emph{Solving the Closest Vector Problem in $2^n$ Time: The Discrete
  Gaussian Strikes Again!}, pp.\ 563--582 in: 2015 IEEE 56th Annual Symposium on Foundations of Computer Science---FOCS 2015, IEEE Computer Soc., 2015.

\bibitem{Aharonov2005a}
D.~Aharonov, O.~Regev,
\emph{Lattice problems in $\mathrm{NP} \cap \mathrm{coNP}$}, 
J. Assoc. Comput. Mach. \textbf{52} (2005), 749--765.

\bibitem{Artmann2017a}
  S.~Artmann, R.~Weismantel, R.~Zenklusen,
\emph{A strongly polynomial algorithm for bimodular integer linear
  programming},
pp.\ 1206--1219 in: STOC'17---Proceedings of the 49th Annual ACM SIGACT Symposium on Theory of Computing, ACM, New York, 2017.
  
\bibitem{Bacher1997a}
  R.~Bacher, P.~de la Harpe, T.~Nagnibeda,
\emph{The lattice of integral flows and the lattice of integral cuts
  on a finite graph},
Bull. Soc. Math. France \textbf{125} (1997), 167--198. 

\bibitem{Bjoerner1993a}
  A.~Bj\"orner, M.~Las Vergnas, B.~Sturmfels, N.~White, G.~Ziegler,
  \emph{Oriented Matroids},
  Cambridge University Press, 1993.

\bibitem{Camion2006a}
P.~Camion,
\emph{Unimodular modules},
Discrete Mathematics \textbf{306} (2006), 2355--2382.

\bibitem{Conway1988a}
J.H.~Conway, N.J.A.~Sloane,
\emph{Sphere Packings, Lattices and Groups},
Springer, 1988.

\bibitem{Conway1992a}
J.H.~Conway, N.J.A.~Sloane, 
\emph{Low-dimensional lattices. VI. Voronoi reduction of
  three-dimensional lattices},
Proc. Roy. Soc. London Ser. A \textbf{436} (1992), 55--68. 

\bibitem{Conway1997a}
J.H.~Conway,
\emph{The sensual (quadratic) form}
(With the assistance of Francis
Y. C. Fung),
Mathematical Association of America, 1997.

\bibitem{Coxeter1962a}
  H.S.M.~Coxeter,
  \emph{The classification of zonohedra by means of projective
    diagrams},
  J. Math. Pure Appl. \textbf{41} (1962), 137--156.

\bibitem{Dinur2003a}
I.~Dinur, G.~Kindler, R.~Raz, S.~Safra,
\emph{Approximating {CVP} to within almost-polynomial factors is {NP}-hard},
Combinatorica \textbf{23} (2003), 205--243.

\bibitem{Ducas2018a}
L.~Ducas, W.P.J.~van Woerden,
\emph{The closest vector problem in tensored root lattices of type A
  and in their duals},
Des. Codes Cryptogr. \textbf{86} (2018), 137--150.

\bibitem{Dutour2009a}
  M.~Dutour Sikiri\'c, A.~Sch\"urmann, F.~Vallentin,
\emph{Complexity and algorithms for computing Voronoi cells of lattices}
Math. Comp. \textbf{78} (2009), 1713--1731. 

\bibitem{vanEmdeBoas1981a}
P.~van Emde Boas,
\emph{Another NP-complete problem and the complexity of computing short vectors
in a lattice}. Tech. rep., University of Amsterdam, Department of
Mathematics, Netherlands. Technical Report 8104.

\bibitem{Erdahl1994a}
R.M.~Erdahl, S.S.~Ryshkov,
\emph{On lattice dicing},
European J. Combin. \textbf{15} (1994), 459--481. 

\bibitem{Erdahl1999a}
  R.M.~Erdahl,
\emph{Zonotopes, dicings, and Voronoi's conjecture on
  parallelohedra},
European J. Combin. \textbf{20} (1999), 527--549. 

\bibitem{Gerritzen1982a}
L.~Gerritzen,
\emph{Die Jacobi-Abbildung \"uber dem Raum der Mumfordkurven},
Math. Ann. \textbf{261} (1982), 81--100.

\bibitem{Goldberg1989a}
A.~Goldberg, R.E.~Tarjan,
\emph{Finding minimum-cost circulations by canceling negative cycles},
J. Assoc. Comput. Mach. \textbf{36} (1989), 873--886. 

\bibitem{Hochbaum1990a}
  D.S.~Hochbaum, J.G.~Shanthikumar,
 \emph{Convex separable optimization is not much harder than linear
   optimization},
 J. Assoc. Comput. Mach. \textbf{37} (1990), 843--862.

\bibitem{Hunkenschroeder2019a}
  C.~Hunkenschr\"oder, G.~Reuland, M.~Schymura,
  \emph{On compact representations of Voronoi cells of
    lattices},
  pp. 261--274 in: Lecture Notes in Comput. Sci., 11480, Springer, 2019.

\bibitem{Jaeger1983a}
F.~Jaeger,
\emph{On space-tiling zonotopes and regular chain-groups},
Ars Combin. \textbf{16} (1983), B, 257--270. 

\bibitem{Karzanov1995a}
A.V.~Karzanov, S.T.~McCormick,
\emph{Polynomial methods for separable convex optimization in unimodular linear spaces with applications},
SODA 1995, 78--87.

\bibitem{Karzanov1997a}
A.V.~Karzanov, S.T.~McCormick,
\emph{Polynomial methods for separable convex optimization in unimodular linear spaces with applications},
SIAM J. Comput. \textbf{26} (1997), 1245--1275. 

\bibitem{Loesch1990a}
H.-F. Loesch,
\emph{Zur Reduktionstheorie von Delone-Voronoi f\"ur matroidische
  quadratische Formen},
Dissertation, Ruhr-Universit\"at Bochum, 1990.

\bibitem{McKilliam2014a}
R.G.~McKilliam, A.~Grant, I.V.~Clarkson,
\emph{Finding a closest point in lattices of Voronoi's first kind},
SIAM J. Discrete Math. \textbf{28} (2014), 1405--1422.

\bibitem{McMullen1975a}
  P. McMullen,
  \emph{Space tiling zonotopes},
  Mathematika \textbf{22} (1975), 202--211.

\bibitem{Micciancio2013a}
  D.~ Micciancio, P.~Voulgaris,
  \emph{A deterministic single exponential time algorithm for most lattice problems based on Voronoi cell computations}, SIAM J. Comput. \textbf{42} (2013), 1364--1391.

\bibitem{Nguyen2009a}
P.Q.~Nguyen, B.~Vall\'ee (eds.).
\emph{The LLL Algorithm --- Survey and Applications},
Springer, 2010.

\bibitem{Oxley2011a}
J.~Oxley,
\emph{Matroid theory (second edition)},
Oxford University Press, 2011.

\bibitem{Rockafellar1970a} 
R.T.~Rockafellar, 
\emph{Convex analysis}, 
Princeton University Press, 1970.

\bibitem{Schrijver1986a}
A.~Schrijver,
\emph{Theory of Linear and Integer Programming},
Wiley, 1986a.

\bibitem{Seymour1980a}
P.D.~Seymour,
\emph{Decomposition of regular matroids},
J. Combin. Theory Ser. B \textbf{28} (1980), 305--359.

\bibitem{Shephard1974a}
  G.C.~Shephard,
  \emph{Space-filling zonotopes},
  Mathematika \textbf{21} (1974), 261--269.

\bibitem{Truemper1992a}
K.~Truemper,
\emph{Matroid decomposition},
Academic Press, 1992.

\bibitem{Tutte1958ab}
W.T.~Tutte,
\emph{A homotopy theorem for matroids, I, II},
Trans. Amer. Math. Soc. \textbf{88} (1958), 144--174.
 
\bibitem{Tutte1965a}
W.T.~Tutte, 
\emph{Lectures on matroids}, 
J. Res. Natl. Bur. Stand. B \textbf{69B} (1965) 1--47.
 
\bibitem{Tutte1971a}
W.T.~Tutte,
\emph{Introduction to the theory of matroids},
American Elsevier Publishing Company, 1971.

\bibitem {Vallentin2000a}
F.~Vallentin.
\emph{\"Uber die Paralleloeder-Vermutung von \textsc{Vorono\"\i}}.
Diploma thesis, University of Dortmund, 2000.

\bibitem{Vallentin2003a}
F.~Vallentin,
\emph{Sphere coverings, lattices, and tilings},
Dissertation, Technische Universit\"at M\"unchen, 2003.

\bibitem{Vallentin2004a}
F.~Vallentin,
\emph{A note on space tiling zonotopes},
arXiv:math/0402053 [math.MG], 2004, 7 pages.

\bibitem{Welsh1976a}
D.J.A.~Welsh,
\emph{Matroid Theory},
Academic Press, 1976.

\end{thebibliography}
\end{document}